\documentclass{article}%
\usepackage{amsmath}
\usepackage{amsfonts}
\usepackage{amssymb}
\usepackage{graphicx}%
\providecommand{\U}[1]{\protect\rule{.1in}{.1in}}
\newtheorem{theorem}{Theorem}

\newtheorem{example}[theorem]{Example}

\newtheorem{lemma}[theorem]{Lemma}

\newtheorem{remark}[theorem]{Remark}

\newenvironment{proof}[1][Proof]{\noindent\textbf{#1.} }{\ \rule{0.5em}{0.5em}}
\begin{document}

\title{Loss of Memory and Convergence of Quantum Markov Processes }
\author{Keiji Matumoto\\National Institute of Informatics, \\Hitotsubashi 2-1-2, Chiyoda-ku, Tokyo 101-8430 Japan}
\maketitle

\section{Abstract}

In a quantum (inhomogeneous) Markov process $\rho_{1}:=\Gamma_{1}\left(
\rho\right)  $, $\,\rho_{2}:=\Gamma_{2}\left(  \rho_{1}\right)  $, $\cdots$,
where $\Gamma_{i}$ are CPTP maps and $\rho$ is the initial state, the the
state of the system is either oscillatory or convergent to a point or
convergent to an oscillatory orbit. Whichever the case it is, "information"
about the initial state is always monotone non-increasing \ and convergent.
This fact motivate us to define an equivalence class of families of quantum
states, which embodies the bundle of all "information quantities" about the
initial state. We show, for \textit{any} quantum inhomogeneous Markov process
over a finite dimensional Hilbert space, the trajectory in the space of the
all equivalence classes is "monotone decreasing" and convergent to a point,
relative to a reasonablly defined topology. Also, a characterization of weak
ergodicity in this picture is given.

\section{Introduction}

A classical (inhomogeneous) Markov process is defined by a sequence $\left\{
P_{i}\right\}  _{i=1}^{\infty}$ of transition probability matrices, and an
initial probability distribution $\boldsymbol{p}$,
\[
\boldsymbol{p}_{1}:=P_{1}\boldsymbol{p}\,,\,\,\boldsymbol{p}_{2}%
:=P_{2}\boldsymbol{p}_{1}\,,\cdots.
\]
A quantum version of (inhomogeneous) Markov process may be defined by a
sequence $\left\{  \Gamma_{i}\right\}  _{i=1}^{\infty}$ of completely positive
and trace preserving (CPTP) maps, and an initial density matrix $\rho$,%
\[
\rho_{1}:=\Gamma_{1}\left(  \rho\right)  ,\,\rho_{2}:=\Gamma_{2}\left(
\rho_{1}\right)  ,\cdots.
\]
(If the probability space is a discrete set, the former is just a special case
of the latter.) A classical or quantum Markov process may converge to a state,
may oscillate, or may asymptotically come close to an oscillating orbit,
depending on eigenvalues of $P$ or $\Gamma$.

Whichever the case it is, "information" about the initial state
($\boldsymbol{p}$ or $\rho$) is non-increasing function of time. This fact may
be described mathematically as follows. Let $\mathcal{E=}\left\{  \rho
_{\theta}\,;\theta\in\Theta\right\}  $ be a family of initial states,
$D\left(  \rho_{\theta_{1}},\rho_{\theta_{2}},\cdots\rho_{\theta_{k}}\right)
$ (an "information quantity") be a positive $k$- points function which is
non-increasing by CPTP maps. ( $D\left(  \rho_{\theta_{1}},\rho_{\theta_{2}%
}\right)  :=\left\Vert \rho_{\theta_{1}}-\rho_{\theta_{2}}\right\Vert _{1}$,
e.g.) \ Also, let
\begin{align*}
\mathcal{E}_{0}  &  :\mathcal{E}=\left\{  \rho_{\theta};\theta\in
\Theta\right\}  ,\\
\,\mathcal{E}_{1}  &  :=\left\{  \rho_{\theta,1};\theta\in\Theta\right\}
,\,\,\rho_{\theta,1}:=\Gamma_{1}\left(  \rho_{\theta}\right)  ,\,\\
\mathcal{E}_{2}  &  :=\left\{  \rho_{\theta,2};\theta\in\Theta\right\}
,\,\,\rho_{\theta,2}:=\Gamma_{2}\left(  \rho_{\theta,1}\right)  ,
\end{align*}
then%
\begin{align*}
D\left(  \rho_{\theta_{1}},\rho_{\theta_{2}},\cdots\rho_{\theta_{k}}\right)
&  \geq D\left(  \rho_{\theta_{1},1},\rho_{\theta_{2},1},\cdots\rho
_{\theta_{k},1}\right) \\
&  \geq D\left(  \rho_{\theta_{1},2},\rho_{\theta_{2},2},\cdots\rho
_{\theta_{k},2}\right) \\
&  \cdots\\
&  \geq0.
\end{align*}
\ \ \ Obviously, the sequence $\left\{  D\left(  \rho_{\theta_{1},i}%
,\cdots\rho_{\theta_{k},i}\right)  \right\}  _{i=1}^{\infty}$ converges, being
monotone decreasing and bounded from below.

So we ask the following question. Is there an object which embodies the
totality of information quantities, which is "monotone decreasing", and
"converges" to a point as time passes? In this paper, as such an object, we
propose an equivalence class of state families over a Hilbert space;
$\mathcal{E}\mathcal{=}\left\{  \rho_{\theta};\theta\in\Theta\right\}  $ is
equivalent to $\mathcal{F=}\left\{  \sigma_{\theta};\theta\in\Theta\right\}  $
\ if and only if
\begin{equation}
D\left(  \rho_{\theta_{1}},\rho_{\theta_{2}},\cdots\rho_{\theta_{k}}\right)
=D\left(  \sigma_{\theta_{1}},\sigma_{\theta_{2}},\cdots\sigma_{\theta_{k}%
}\right)  , \label{Drho=Dsigma}%
\end{equation}
holds for any CPTP monotone decreasing functional $D$. Also, we introduce
order structure in the space of these equivalence classes; $\left[
\mathcal{E}\right]  \succeq\left[  \mathcal{F}\right]  $ if and only if
\begin{equation}
D\left(  \rho_{\theta_{1}},\rho_{\theta_{2}},\cdots\rho_{\theta_{k}}\right)
\geq D\left(  \sigma_{\theta_{1}},\sigma_{\theta_{2}},\cdots\sigma_{\theta
_{k}}\right)  \label{Drho>Dsigma}%
\end{equation}
holds for any $k$-point functional $D$ which is monotone decreasing by CPTP maps.

Obviously, the sequence $\left\{  \left[  \mathcal{E}_{i}\right]  \right\}
_{i=0}^{\infty}$ is monotone decreasing
\[
\left[  \mathcal{E}\right]  =\left[  \mathcal{E}_{0}\right]  \succeq\left[
\mathcal{E}_{1}\right]  \succeq\cdots,
\]
and the value of each $D$ is convergent. But to make above rough statement
rigorous, we have to prove the existence of the family
\[
\mathcal{E}_{\infty}:=\left\{  \rho_{\theta,\infty};\theta\in\Theta\right\}
,
\]
such that
\begin{equation}
\lim_{i\rightarrow\infty}D\left(  \rho_{\theta_{1},i},\rho_{\theta_{2}%
,i},\cdots\rho_{\theta_{k},i}\right)  =D\left(  \rho_{\theta_{1},\infty}%
,\rho_{\theta_{2},\infty},\cdots\rho_{\theta_{k},\infty}\right)
\label{limD=D}%
\end{equation}
holds for any well-behaved functional $D$, and that
\[
\lim_{i\rightarrow\infty}\left[  \mathcal{E}_{i}\right]  =\left[
\mathcal{E}_{\infty}\right]
\]
holds with respect to a reasonably defined topology.

The line of arguments in this paper is more or less in reminiscent of
\cite{lindqvist:77}. However, there are some notable differences. First,
\cite{lindqvist:77} is dealing with classical Markov processes (over the
finite set), while we are dealing with its quantum counterpart. Second, in
\cite{lindqvist:77}, $\Theta$ is a finite set; an initial state is
concentrated at one of the site. Due to these two, \cite{lindqvist:77} can
utilize Blackwell measure \cite{Torgersen}, for which there is one-to-one
correspondence with an equivalence class of families of probability
distributions over \textit{all} measurable spaces. In quantum case, however,
the counterpart of Blackwell measure so far proposed is a state over a very
huge algebra \cite{GutaJencova}, and thus not handy to deal with. Hence, we
prefer to treat the equivalence classes directly, rather than using the
quantum version of Blackwell measure.

\section{Equivalent classes of finite dimensional state families}

\label{sec:equivalence-class}

Let $\mathcal{B}\left(  \mathcal{H}\right)  $ and $\mathcal{S}\left(
\mathcal{H}\right)  $ be the set of operators and density operators over
$\mathcal{H}$, respectively. Let $\mathcal{C}\left(  \mathcal{H}\right)  $
denote CPTP maps from $\mathcal{B}\left(  \mathcal{H}\right)  $ to itself. Let
$\mathcal{H}:=%
%TCIMACRO{\U{2102} }%
%BeginExpansion
\mathbb{C}
%EndExpansion
^{d}$, and $\Theta$ be a set. Denote by $\mathcal{S}\left(  \mathcal{H}%
\right)  ^{\Theta}$ the set of all families of states in $\mathcal{H}$
parameterized by $\theta\in\Theta$.

Let Introduce preorder $\succeq$ to $\mathcal{S}\left(  \mathcal{H}\right)
^{\Theta}$ : \ Given $\mathcal{E}:=\left\{  \rho_{\theta};\theta\in
\Theta\right\}  $, $\mathcal{F}:=\left\{  \sigma_{\theta};\theta\in
\Theta\right\}  \in\mathcal{S}\left(  \mathcal{H}\right)  ^{\Theta}$, we write
$\mathcal{E\succeq F}$ if and only if
\begin{equation}
\Lambda\left(  \mathcal{E}\right)  =\mathcal{F}\,,\ \exists\Lambda
\in\mathcal{C}\left(  \mathcal{H}\right)  ,\label{randomization}%
\end{equation}
with
\[
\Lambda\left(  \mathcal{E}\right)  :=\left\{  \Lambda\left(  \rho_{\theta
}\right)  ;\theta\in\Theta\right\}  .
\]
(\ref{randomization}) holds if and only if (\ref{Drho>Dsigma}) holds for any
$k$-point CPTP monotone non-increasing functional $D$ with (\ref{D-D<f}) and
for any $k$ \cite{Matsumoto}. Thus, definition here is the same as the one
mentioned in the introduction.

Introduce equivalence relation $\equiv$ in $\mathcal{S}\left(  \mathcal{H}%
\right)  ^{\Theta}$as follows:
\begin{equation}
\mathcal{E}\equiv\mathcal{F}\Leftrightarrow\mathcal{E\succeq F}%
,\,\mathcal{F\succeq E\,}.\, \label{def-equiv}%
\end{equation}
We denote by $\mathbb{E}\left(  \Theta,\mathcal{H}\right)  $ the totality of
this equivalence classes. $\left[  \mathcal{E}\right]  $ denotes the
equivalence class to which $\mathcal{E}$ belongs.

Introduce pseudo metric $\Delta$ on $\mathcal{S}\left(  \mathcal{H}\right)
^{\Theta}$ as follows:%
\begin{align*}
\Delta\left(  \mathcal{E},\mathcal{F}\right)   &  :=\max\left\{  \delta\left(
\mathcal{E},\mathcal{F}\right)  ,\delta\left(  \mathcal{F},\mathcal{E}\right)
\right\}  ,\\
\delta\left(  \mathcal{E},\mathcal{F}\right)   &  :=\inf_{\Lambda_{2}%
\in\mathcal{C}\left(  \mathcal{H}\right)  }\sup_{\theta\in\Theta}\left\Vert
\Lambda\left(  \rho_{\theta}\right)  -\sigma_{\theta}\right\Vert _{1},
\end{align*}
where $\left\Vert A\right\Vert _{1}=\mathrm{tr}\,\sqrt{A^{\ast}A}$. Observe,
by (\ref{randomization}),
\begin{align}
\Delta\left(  \mathcal{E},\mathcal{F}\right)   &  =0\Leftrightarrow
\mathcal{E}\equiv\mathcal{F},\label{Delta=0-equiv}\\
\mathcal{E}  &  \equiv\mathcal{E}^{\prime},\mathcal{F}\equiv\mathcal{F}%
^{\prime}\Rightarrow\delta\left(  \mathcal{E},\mathcal{F}\right)
=\delta\left(  \mathcal{E}^{\prime},\mathcal{F}^{\prime}\right)  .
\label{equiv->delta=delta}%
\end{align}
Therefore, each of $\delta$ and $\Delta$ naturally defines a two point
functional in $\mathbb{E}\left(  \Theta,\mathcal{H}\right)  $, which is also
denoted by $\delta$ and $\Delta$:
\begin{align*}
\delta\left(  \left[  \mathcal{E}\right]  ,\left[  \mathcal{F}\right]
\right)   &  :=\delta\left(  \mathcal{E},\mathcal{F}\right)  ,\\
\Delta\left(  \left[  \mathcal{E}\right]  ,\left[  \mathcal{F}\right]
\right)   &  :=\Delta\left(  \mathcal{E},\mathcal{F}\right)  .
\end{align*}
The topology over $\mathbb{E}\left(  \Theta,\mathcal{H}\right)  $ indeed by
the metric $\Delta$ is called \textit{strong topology}.

By definition, we have%
\begin{align}
\delta\left(  \mathcal{E},\mathcal{F}\right)   &  \leq\delta\left(
\mathcal{E},\mathcal{E}^{\prime}\right)  +\delta\left(  \mathcal{E}^{\prime
},\mathcal{F}\right)  ,\label{delta<delta+delta}\\
\delta\left(  \mathcal{E},\mathcal{F}\right)   &  =0\Leftrightarrow
\,\mathcal{E}\succeq\mathcal{F}. \label{L(r)=s-2}%
\end{align}
and
\begin{align}
\Delta\left(  \mathcal{E},\mathcal{F}\right)   &  \leq\Delta\left(
\mathcal{E},\mathcal{E}^{\prime}\right)  +\Delta\left(  \mathcal{E}^{\prime
},\mathcal{F}\right)  ,\label{d<d+d}\\
\Delta\left(  \mathcal{E},\mathcal{F}\right)   &  \geq\Delta\left(
\Lambda\left(  \mathcal{E}\right)  ,\,\Lambda\left(  \mathcal{F}\right)
\right)  \,,\,\forall\Lambda\in\mathcal{C}\left(  \mathcal{H}\right)
\label{d>dL}%
\end{align}

Define projection from $\mathcal{S}\left(  \mathcal{H}\right)  ^{\Theta\text{
}}$ to $\mathbb{E}\left(  \Theta,\mathcal{H}\right)  $ such that
\[
P:\mathcal{E\rightarrow}\left[  \mathcal{E}\right]  .
\]

\begin{lemma}
\label{lem:s-compact}Suppose that $\Theta$ is a finite set. Then,
$\mathbb{E}\left(  \Theta,\mathcal{H}\right)  $ is compact with respect to
strong topology.
\end{lemma}

\begin{proof}
Define a norm
\[
\left\Vert \left\{  X_{\theta};\theta\in\Theta\right\}  \right\Vert _{1}%
:=\max_{\theta\in\Theta}\left\Vert X_{\theta}\right\Vert _{1}\,,
\]
where $\left\{  X_{\theta};\theta\in\Theta\right\}  \in\mathcal{B}\left(
\mathcal{H}\right)  ^{\Theta}$, and equip $\mathcal{B}\left(  \mathcal{H}%
\right)  ^{\Theta}$ with the topology defined by this norm. Then, since
$\mathcal{B}\left(  \mathcal{H}\right)  ^{\Theta}$ is finite dimensional
vector space, all the norm are topologically equivalent. Thus, $\mathcal{S}%
\left(  \mathcal{H}\right)  ^{\Theta}$ is compact with respect to the topology
defined above. Also, as is shown below, the projection $P$ from $\mathcal{S}%
\left(  \mathcal{H}\right)  ^{\Theta}$ onto $\mathbb{E}\left(  \Theta
,\mathcal{H}\right)  $ is continuous. Therefore, $\mathbb{E}\left(
\Theta,\mathcal{H}\right)  $ is compact.

Continuity of $P$ is proved as follows. Denote
\[
\mathcal{E}-\mathcal{F}:=\left\{  \rho_{\theta}-\sigma_{\theta};\theta
\in\Theta\right\}  ,
\]
where $\mathcal{E}:=\left\{  \rho_{\theta};\theta\in\Theta\right\}  $ and
$\mathcal{F}:=\left\{  \sigma_{\theta};\theta\in\Theta\right\}  $ .\ Observe
\begin{equation}
\left\Vert \mathcal{E}-\mathcal{F}\right\Vert _{1}\geq\Delta\left(  \left[
\mathcal{E}\right]  ,\left[  \mathcal{F}\right]  \right)  .\label{Delta>Delta}%
\end{equation}
Also observe, for any point $\left[  \mathcal{E}\right]  \in\mathbb{E}\left(
\Theta,\mathcal{H}\right)  $ in an open set $O\subset\mathbb{E}\left(
\Theta,\mathcal{H}\right)  $, there is $\varepsilon$ with \ \
\[
\left\{  \left[  \mathcal{F}\,\right]  ;\,\Delta\left(  \left[  \mathcal{E}%
\right]  ,\left[  \mathcal{F}\right]  \right)  <\varepsilon\right\}  \subset
O.
\]
Therefore, by (\ref{Delta>Delta}), 5%
\begin{align*}
P^{-1}\left(  O\right)   &  \supset P^{-1}\left(  \left\{  \left[
\mathcal{F}\,\right]  ;\,\Delta\left(  \left[  \mathcal{E}\right]  ,\left[
\mathcal{F}\right]  \right)  <\varepsilon\right\}  \right)  ,\\
&  =\left\{  \mathcal{F};\,\Delta\left(  \mathcal{E},\mathcal{F}\right)
<\varepsilon\mathcal{\,}\right\}  ,\\
&  \supset\left\{  \mathcal{F};\left\Vert \mathcal{E}-\mathcal{F}\right\Vert
_{1}<\varepsilon\mathcal{\,}\right\}  ,
\end{align*}
which means $P^{-1}\left(  O\right)  $ is open. Therefore, $P$ is continuous.
\end{proof}

\begin{remark}
It is may be worthwhile to mention that the partial order "$\mathcal{\succeq}$
" has a good operational meaning. That is,  $\mathcal{E}\succeq\mathcal{F}$
holds if and only if, for any task defined on the parameter set $\Theta$, the
optimal gain is always larger in $\mathcal{E}$ tha in $\mathcal{F}$
\cite{Matsumoto}. 
\end{remark}

\section{Convergence of sequences of equivalence classes}

Given a sequence of CPTP maps $\Gamma_{i}$ $\in\mathcal{C}\left(
\mathcal{H}\right)  $ ($i=1,2,\cdots$), define recursively,
\begin{align*}
\mathcal{E}_{0}  &  :\mathcal{=E}=\left\{  \rho_{\theta};\theta\in
\Theta\right\}  ,\\
\,\mathcal{E}_{1}  &  :=\left\{  \rho_{\theta,1};\theta\in\Theta\right\}
,\,\,\rho_{\theta,1}:=\Gamma_{1}\left(  \rho_{\theta}\right)  ,\,\\
\mathcal{E}_{2}  &  :=\left\{  \rho_{\theta,2};\theta\in\Theta\right\}
,\,\,\rho_{\theta,2}:=\Gamma_{2}\left(  \rho_{\theta,1}\right)  ,
\end{align*}
and so on, and consider the sequence $\left\{  \left[  \mathcal{E}_{i}\right]
\right\}  _{i=0}^{\infty}$ .

\begin{theorem}
Let $\left\{  \left[  \mathcal{E}_{i}\right]  \right\}  _{i=0}^{\infty}$ be
defined as above, and $\Theta$ be any set. Then, there is $\mathcal{E}%
_{\infty}=\left\{  \rho_{\theta,\infty};\theta\in\Theta\right\}  $ such that
\begin{align}
\lim_{i\rightarrow\infty}\Delta\left(  \left[  \mathcal{E}_{i}\right]
,\left[  \mathcal{E}_{\infty}\right]  \right)   &  =0,\label{converge}\\
\mathcal{E}  &  \mathcal{=E}_{0}\succeq\mathcal{E}_{1}\succeq\cdots
\succeq\mathcal{E}_{\infty},\label{larger}\\
\Delta\left(  \left[  \mathcal{E}_{i_{1}}\right]  ,\left[  \mathcal{E}%
_{\infty}\right]  \right)   &  \geq\Delta\left(  \left[  \mathcal{E}_{i_{2}%
}\right]  ,\left[  \mathcal{E}_{\infty}\right]  \right)  \,,\,\,i_{1}\leq
i_{2} \label{d>d}%
\end{align}

\end{theorem}

\begin{proof}
This proof is much draws upon the one of Lemma\thinspace2.1\ of
\cite{lindqvist:77}. Let $\Xi$ be a set with $\left\vert \Xi\right\vert
=d^{2}$, and $\left\{  \rho^{\xi};\xi\in\Xi\right\}  $ be a basis of the space
of Hermitian operators over $\mathcal{H}$ viewed as a real vector space. Then,
define
\begin{align*}
\mathcal{\tilde{E}}_{0}  &  :=\left\{  \rho^{\xi};\xi\in\Xi\right\}  ,\\
\,\mathcal{\tilde{E}}_{1}  &  :=\left\{  \rho_{1}^{\xi};\xi\in\Xi\right\}
,\,\,\rho_{1}^{\xi}:=\Gamma_{1}\left(  \rho^{\xi}\right)  ,\,\\
\mathcal{\tilde{E}}_{2}  &  :=\left\{  \rho_{2}^{\xi};\xi\in\Xi\right\}
,\,\,\rho_{2}^{\xi}:=\Gamma_{2}\left(  \rho_{1}^{\xi}\right)  ,
\end{align*}
and so on, and consider the sequence of equivalence classes of state families
$\left\{  \left[  \mathcal{\tilde{E}}_{i}\right]  \right\}  _{i=0}^{\infty}$ ,
where the equivalence class is defined by (\ref{def-equiv}).

Due to Lemmas\thinspace\ref{lem:s-compact} and \ref{lem:compact-accumulate},
there is an accumulation point of the set $\left\{  \left[  \mathcal{\tilde
{E}}_{i}\right]  \right\}  _{i=0}^{\infty}$. Let that accumulation point be
$\left[  \mathcal{\tilde{E}}_{\infty}\right]  $, where
\[
\mathcal{\tilde{E}}_{\infty}=\left\{  \rho_{\infty}^{\xi}\,;\xi\in\Xi\right\}
.
\]
Since $\mathbb{E}\left(  \Xi,\mathcal{H}\right)  $ is topologized by the
topology based on metric $\Delta$, it satisfies the first axiom of
countability, due to \ Lemma\thinspace\ref{lem:metric-countable}. Therefore,
by Lemma\thinspace\ref{lem:cluster}, there is a subsequence $\left\{
n_{i}\right\}  _{i=1}^{\infty}$ such that
\begin{equation}
\lim_{i\rightarrow\infty}\Delta\left(  \left[  \mathcal{\tilde{E}}_{n_{i}%
}\right]  ,\,\left[  \mathcal{\tilde{E}}_{\infty}\right]  \right)  =0.
\label{Delta->0}%
\end{equation}

Since $\left\{  \rho^{\xi};;\xi\in\Xi\right\}  $ is a basis of $\mathcal{S}%
\left(  \mathcal{H}\right)  $, there are real valued functions $\alpha
_{\theta,\xi}$ with
\[
\rho_{\theta}=\sum_{\xi\in\Xi}\alpha_{\theta,\xi}\rho^{\xi},
\]
and $\left\{  \tilde{\rho}^{\xi};\,\xi\in\Xi\right\}  $ be the dual base,
\[
\alpha_{\theta,\xi}=\mathrm{tr}\,\rho_{\theta}\tilde{\rho}_{\xi}.
\]
Then,
\[
\left\vert \alpha_{\theta,\xi}\right\vert \leq\left\Vert \rho_{\theta
}\right\Vert _{1}\left\Vert \tilde{\rho}_{\xi}\right\Vert =\left\Vert
\tilde{\rho}_{\xi}\right\Vert .\,
\]
Since $\Xi$ is a finite set,%
\begin{equation}
\sup_{\theta\in\Theta,\xi\in\Xi}\left\vert \alpha_{\theta,\xi}\right\vert
<\infty. \label{sup a <infty}%
\end{equation}
Since $\Gamma_{i}$ ($i=1,2,\cdots$) are linear, \
\[
\rho_{\theta,i}=\sum_{\xi\in\Xi}\alpha_{\theta,\xi}\rho_{i}^{\xi}.
\]
Define
\[
\rho_{\theta,\infty}:=\sum_{\xi\in\Xi}\alpha_{\theta,\xi}\rho_{\infty}^{\xi}.
\]

Observe, by (\ref{L(r)=s-2}), if $i_{1}\leq i_{2}$,
\begin{equation}
\delta\left(  \mathcal{E}_{i_{1}},\,\mathcal{E}_{i_{2}}\right)  =0.
\label{delta=0-2}%
\end{equation}
Therefore, by (\ref{delta<delta+delta}),
\begin{align}
\delta\left(  \mathcal{E}_{\infty},\mathcal{E}_{i_{2}}\,\right)   &
\leq\delta\left(  \mathcal{E}_{\infty},\mathcal{E}_{i_{1}}\,\right)
+\delta\left(  \mathcal{E}_{i_{1}},\mathcal{E}_{i_{2}}\right) \nonumber\\
&  =\delta\left(  \mathcal{E}_{\infty},\mathcal{E}_{i_{1}}\,\right)  .
\label{delta<delta}%
\end{align}

Therefore, by choosing $j$ so that $n_{j}\leq i$, we have%
\begin{align*}
\delta\left(  \mathcal{E}_{\infty},\mathcal{E}_{i}\,\right)   &  \leq
\delta\left(  \mathcal{E}_{\infty},\mathcal{E}_{n_{j}}\,\right) \\
&  =\inf_{\Lambda\in\mathcal{C}\left(  \mathcal{H}\right)  }\sup_{\theta
\in\Theta}\left\Vert \Lambda\left(  \sum_{\xi\in\Xi}\alpha_{\theta,\xi}%
\rho_{\infty}^{\xi}\right)  -\sum_{\xi\in\Xi}\alpha_{\theta,\xi}\rho_{n_{j}%
}^{\xi}\right\Vert _{1}\\
&  \leq\inf_{\Lambda\in\mathcal{C}\left(  \mathcal{H}\right)  }\sup_{\theta
\in\Theta}\sum_{\xi\in\Xi}\left\vert \alpha_{\theta,\xi}\right\vert \left\Vert
\Lambda\left(  \rho_{\infty}^{\xi}\right)  -\rho_{n_{j}}^{\xi}\right\Vert
_{1}\\
&  \leq d^{2}\sup_{\theta\in\Theta,\xi\in\Xi}\left\vert \alpha_{\theta,\xi
}\right\vert \inf_{\Lambda\in\mathcal{C}\left(  \mathcal{H}\right)  }\sup
_{\xi\in\Xi}\left\Vert \Lambda\left(  \rho_{\infty}^{\xi}\right)  -\rho
_{n_{j}}^{\xi}\right\Vert _{1}\\
&  =d^{2}\sup_{\theta\in\Theta,\xi\in\Xi}\left\vert \alpha_{\theta,\xi
}\right\vert \,\delta\left(  \,\mathcal{\tilde{E}}_{\infty},\mathcal{\tilde
{E}}_{n_{j}}\right)  ,
\end{align*}
which, combined with (\ref{Delta->0}) and (\ref{sup a <infty}), leads to
\begin{align}
\lim_{i\rightarrow\infty}\delta\left(  \mathcal{E}_{\infty},\mathcal{E}%
_{i}\,\right)   &  \leq d^{2}\sup_{\theta\in\Theta,\xi\in\Xi}\left\vert
\alpha_{\theta,\xi}\right\vert \,\lim_{n_{j}\rightarrow\infty}\delta\left(
\,\mathcal{\tilde{E}}_{\infty},\mathcal{\tilde{E}}_{n_{j}}\right) \nonumber\\
&  =0. \label{delta->0}%
\end{align}
Similarly, for any $i$, taking $j$ large so that $n_{j}\geq i$ holds, we have
\begin{align*}
\delta\left(  \mathcal{E}_{i},\,\mathcal{E}_{\infty}\right)   &  \leq
\delta\left(  \mathcal{E}_{i},\,\mathcal{E}_{n_{j}}\right)  +\delta\left(
\mathcal{E}_{n_{j}},\,\mathcal{E}_{\infty}\right) \\
&  =\delta\left(  \mathcal{E}_{n_{j}},\,\mathcal{E}_{\infty}\right) \\
&  =\inf_{\Lambda\in\mathcal{C}\left(  \mathcal{H}\right)  }\sup_{\theta
\in\Theta}\left\Vert \Lambda\left(  \sum_{\xi\in\Xi}\alpha_{\theta,\xi}%
\rho_{n_{j}}^{\xi}\right)  -\sum_{\xi\in\Xi}\alpha_{\theta,\xi}\rho_{\infty
}^{\xi}\right\Vert _{1}\\
&  \leq d^{2}\sup_{\theta\in\Theta,\xi\in\Xi}\left\vert \alpha_{\theta,\xi
}\right\vert \,\delta\left(  \mathcal{\tilde{E}}_{n_{j}},\,\mathcal{\tilde{E}%
}_{\infty}\right)  ,
\end{align*}
which, with the help of (\ref{Delta->0}) and (\ref{sup a <infty}), leads to%
\begin{align}
\delta\left(  \mathcal{E}_{i},\,\mathcal{E}_{\infty}\right)   &  \leq
d^{2}\sup_{\theta\in\Theta,\xi\in\Xi}\left\vert \alpha_{\theta,\xi}\right\vert
\,\lim_{j\rightarrow\infty}\,\delta\left(  \mathcal{\tilde{E}}_{n_{j}%
},\,\mathcal{\tilde{E}}_{\infty}\right) \nonumber\\
&  =0. \label{delta=0}%
\end{align}

Combining (\ref{delta->0}) and (\ref{delta=0}) leads to (\ref{converge}).
(\ref{delta=0-2}) and (\ref{delta=0}) implies (\ref{larger}).
(\ref{delta<delta}) and (\ref{delta=0}) leads to (\ref{d>d}).
\end{proof}

We say a quantum Markov process is \textit{weakly ergodic} if and only if the
state tends to be independent of the initial state, or
\[
\lim_{i\rightarrow\infty}\sup_{\rho,\rho^{\prime}}\left\Vert \Gamma_{i}%
\circ\cdots\circ\Gamma_{2}\circ\Gamma_{1}\left(  \rho\right)  -\Gamma_{i}%
\circ\cdots\circ\Gamma_{2}\circ\Gamma_{1}\left(  \rho^{\prime}\right)
\right\Vert _{1}=0.
\]
Weak ergodicity, by definition, is equivalent to the convergence to one-point
family $\mathcal{E}_{\ast}:=\left\{  \rho_{\ast};\theta\in\Theta\right\}  \,$.
This means the information about the initial state is completely lost.

\begin{theorem}
\label{th:ergodic}A quantum Markov process is weakly ergodic if and only if
$\mathcal{E}_{\ast}=\mathcal{E}_{\infty}$.
\end{theorem}

\begin{proof}
Suppose weaky ergodicity holds. Fix $\theta_{0}\in\Theta$ , and let
$\Lambda_{i}$ be a CPTP map such that $\Lambda_{i}\left(  \rho_{\ast}\right)
=\rho_{\theta_{0},i}$. Then,
\begin{align*}
\delta\left(  \mathcal{E}_{\ast},\mathcal{E}_{i}\right)   &  =\inf_{\Lambda
}\left\Vert \Lambda\left(  \rho_{\ast}\right)  -\rho_{\theta,i}\right\Vert
_{1}\\
&  \leq\left\Vert \Lambda_{i}\left(  \rho_{\ast}\right)  -\rho_{\theta
,i}\right\Vert _{1}=\left\Vert \rho_{\theta_{0},i}-\rho_{\theta,i}\right\Vert
_{1}\rightarrow0.
\end{align*}
Define a CPTP map $\Lambda^{\prime}$ such that
\[
\Lambda^{\prime}\left(  \rho\right)  =\rho_{\ast}\,,\,\forall\rho
\in\mathcal{S}\left(  \mathcal{H}\right)  .
\]
Then
\begin{align*}
\delta\left(  \mathcal{E}_{i},\mathcal{E}_{\ast}\right)   &  =\inf_{\Lambda
}\left\Vert \Lambda\left(  \rho_{\theta,i}\right)  -\rho_{\ast}\right\Vert
_{1}\\
&  \leq\left\Vert \Lambda^{\prime}\left(  \rho_{\theta,i}\right)  -\rho_{\ast
}\right\Vert _{1}=0.
\end{align*}
Therefore, we have
\[
\Delta\left(  \mathcal{E}_{\ast},\mathcal{E}_{i}\right)  \rightarrow0.
\]
Since the strong topology is based on the distance $\Delta$, it is a Hausdorff
space. Therefore, any sequence has at most one convergent point. Hence,
$\mathcal{E}_{\ast}=\mathcal{E}_{\infty}$.

Conversely, suppose $\mathcal{E}_{\ast}=\mathcal{E}_{\infty}$. Then,
\begin{align*}
\sup_{\theta,\theta^{\prime}\in\Theta}\left\Vert \rho_{\theta,i}-\rho
_{\theta^{\prime},i}\right\Vert _{1}  &  \leq\sup_{\theta,\theta^{\prime}%
\in\Theta}\left(  \left\Vert \Lambda\left(  \rho_{\ast}\right)  -\rho
_{\theta,i}\right\Vert _{1}+\left\Vert \Lambda\left(  \rho_{\ast}\right)
-\rho_{\theta^{\prime},i}\right\Vert _{1}\right) \\
&  =2\sup_{\theta\in\Theta}\left\Vert \Lambda\left(  \rho_{\ast}\right)
-\rho_{\theta,i}\right\Vert _{1}%
\end{align*}
Since this holds for any $\Lambda$, we have
\begin{align*}
\sup_{\theta,\theta^{\prime}\in\Theta}\left\Vert \rho_{\theta,i}-\rho
_{\theta^{\prime},i}\right\Vert _{1}  &  \leq2\inf_{\Lambda}\sup_{\theta
\in\Theta}\left\Vert \Lambda\left(  \rho_{\ast}\right)  -\rho_{\theta
,i}\right\Vert _{1}\\
&  \leq2\Delta\left(  \mathcal{E}_{\infty},\mathcal{E}_{i}\right)
\rightarrow0.
\end{align*}
Letting $\mathcal{E}=\mathcal{S}\left(  \mathcal{H}\right)  $, we have weak
ergodicity. Thus the proof is complete.
\end{proof}

Finally, we show $\left[  \mathcal{E}_{\infty}\right]  $ is a fixed point if
the Markov process is homogeneous.

\begin{theorem}
\label{th:fixed-point}Suppose $\Gamma_{i}=\Gamma$ ($i=1,2,\cdots$). Then,
\[
\Gamma\left(  \mathcal{E}_{\infty}\right)  \equiv\mathcal{E}_{\infty}.
\]

\end{theorem}

\begin{proof}
Suppose
\[
\Delta\left(  \Gamma\left(  \mathcal{E}_{\infty}\right)  ,\mathcal{E}_{\infty
}\right)  =c\geq0.
\]
Choose $i$ so that
\begin{equation}
\Delta\left(  \mathcal{E}_{i},\mathcal{E}_{\infty}\right)  \leq\frac{c}{3}.
\label{delta<c/3}%
\end{equation}
By (\ref{d>d}),%
\begin{equation}
\Delta\left(  \mathcal{E}_{i+1},\mathcal{E}_{\infty}\right)  \leq\frac{c}{3}.
\label{delta<c/3-2}%
\end{equation}
Then,
\begin{align*}
\Delta\left(  \mathcal{E}_{i},\mathcal{E}_{\infty}\right)   &  \geq
\Delta\left(  \Gamma\left(  \mathcal{E}_{i}\right)  ,\Gamma\left(
\mathcal{E}_{\infty}\right)  \right) \\
&  =\Delta\left(  \mathcal{E}_{i+1},\Gamma\left(  \mathcal{E}_{\infty}\right)
\right) \\
&  \geq\Delta\left(  \Gamma\left(  \mathcal{E}_{\infty}\right)  ,\mathcal{E}%
_{\infty}\right)  -\Delta\left(  \mathcal{E}_{i+1},\mathcal{E}_{\infty}\right)
\\
&  \geq c-c/3=2c/3,
\end{align*}
where the inequality in the first line is by (\ref{d>dL}), the one in the
third line is by (\ref{d<d+d}), and the one in the fourth line is by
(\ref{delta<c/3-2}).

This, combined with (\ref{delta<c/3}), leads to
\[
\Delta\left(  \Gamma\left(  \mathcal{E}_{\infty}\right)  ,\mathcal{E}_{\infty
}\right)  =c=0.
\]
Thus we have the theorem.
\end{proof}

\section{Limits of information quantities}

\begin{theorem}
\label{th:dDconverge}Consider a $k$-point function
\[
D:\mathcal{S}\left(  \mathcal{H}\right)  \times\mathcal{S}\left(
\mathcal{H}\right)  \times\cdots\times\mathcal{S}\left(  \mathcal{H}\right)
\rightarrow%
%TCIMACRO{\U{211d} }%
%BeginExpansion
\mathbb{R}
%EndExpansion
_{+}\cup\left\{  0\right\}
\]
which is monotone decreasing by CPTP maps. Suppose that
\begin{align}
&  \left\vert D\left(  X_{1},X_{2},\cdots,X_{k}\right)  -D\left(  Y_{1}%
,Y_{2},\cdots,Y_{k}\right)  \right\vert \nonumber\\
&  \leq f\left(  \left\Vert X_{1}-Y_{1}\right\Vert _{1},\left\Vert X_{2}%
-Y_{2}\right\Vert _{1},\cdots,\left\Vert X_{k}-Y_{k}\right\Vert _{1}\right)
\label{D-D<f}%
\end{align}
holds for any $X_{j},Y_{j}\in\mathcal{S}\left(  \mathcal{H}\right)  $
($j=1,2,\cdots,k$), with $f$ being continuous and
\begin{equation}
f\left(  0,0,\cdots,0\right)  =0. \label{f=0}%
\end{equation}
Then we have
\begin{equation}
\lim_{i\rightarrow\infty}D\left(  \rho_{\theta_{1},i},\cdots,\rho_{\theta
_{k},i}\right)  =D\left(  \rho_{\theta_{1},\infty},\cdots,\rho_{\theta
_{k},\infty}\right)  . \label{D->D}%
\end{equation}

\end{theorem}

\begin{proof}
Due to monotonicity and positivity of $D$, the sequence $\left\{  D\left(
\rho_{\theta_{1},,i},\cdots,\rho_{\theta_{k},i}\right)  \right\}
_{i=1}^{\infty}$ is monotone decreasing and bounded from below. Therefore,
this sequence converges.

By (\ref{larger}),
\begin{equation}
D\left(  \rho_{\theta_{1},i},\cdots,\rho_{\theta_{k},i}\right)  \geq D\left(
\rho_{\theta_{1},\infty},\cdots,\rho_{\theta_{k},\infty}\right)  . \label{D>D}%
\end{equation}
Also,
\begin{align*}
&  \lim_{i\rightarrow\infty}\left\{  D\left(  \rho_{\theta_{1},i},\cdots
,\rho_{\theta_{k},i}\right)  -D\left(  \rho_{\theta_{1},\infty},\cdots
,\rho_{\theta_{k},\infty}\right)  \right\} \\
&  \leq\lim_{i\rightarrow\infty}\inf_{\Lambda\in\mathcal{C}\left(
\mathcal{H}\right)  }\left\{  D\left(  \rho_{\theta_{1},i},\cdots,\rho
_{\theta_{k},i}\right)  -D\left(  \Lambda\left(  \rho_{\theta_{1},\infty
}\right)  ,\cdots,\Lambda\left(  \rho_{\theta_{k},\infty}\right)  \right)
\right\} \\
&  \leq\lim_{i\rightarrow\infty}\inf_{\Lambda\in\mathcal{C}\left(
\mathcal{H}\right)  }f\left(  \left\Vert \rho_{\theta_{1},i}-\Lambda\left(
\rho_{\theta_{1},\infty}\right)  \right\Vert _{1},\cdots,\left\Vert
\rho_{\theta_{k},i}-\Lambda\left(  \rho_{\theta_{k},\infty}\right)
\right\Vert _{1}\right) \\
&  \leq\lim_{i\rightarrow\infty}f\left(  \inf_{\Lambda\in\mathcal{C}\left(
\mathcal{H}\right)  }\left\Vert \rho_{\theta_{1},i}-\Lambda\left(
\rho_{\theta_{1},\infty}\right)  \right\Vert _{1},\cdots,\inf_{\Lambda
\in\mathcal{C}\left(  \mathcal{H}\right)  }\left\Vert \rho_{\theta_{k}%
,i}-\Lambda\left(  \rho_{\theta_{k},\infty}\right)  \right\Vert _{1}\right) \\
&  =f\left(  \lim_{i\rightarrow\infty}\inf_{\Lambda\in\mathcal{C}\left(
\mathcal{H}\right)  }\left\Vert \rho_{\theta_{1},i}-\Lambda\left(
\rho_{\theta_{1},\infty}\right)  \right\Vert _{1},\cdots,\lim_{i\rightarrow
\infty}\inf_{\Lambda\in\mathcal{C}\left(  \mathcal{H}\right)  }\left\Vert
\rho_{\theta_{k},i}-\Lambda\left(  \rho_{\theta_{k},\infty}\right)
\right\Vert _{1}\right) \\
&  =0,
\end{align*}
where the first inequality is by the fact that $D$ is monotone decreasing by
CPTP maps, the one in the third line is by (\ref{D-D<f}), the identity in the
fifth line is due to continuity of $f$, and the inequality in the sixth line
is by (\ref{converge}) and (\ref{f=0}). Therefore,
\begin{equation}
\lim_{i\rightarrow\infty}D\left(  \rho_{\theta_{1},i},\cdots,\rho_{\theta
_{k},i}\right)  \leq D\left(  \rho_{\theta_{1},,\infty},\cdots,\rho
_{\theta_{k},\infty}\right)  . \label{D<D}%
\end{equation}
Combining (\ref{D>D}) and (\ref{D<D}), we obtain (\ref{D->D}).
\end{proof}

\begin{example}
With
\[
D\left(  \rho_{1},\rho_{2}\right)  :=\left\Vert \rho_{1}-\rho_{2}\right\Vert
_{1},
\]
the premise of Theorem\thinspace\ref{th:dDconverge} is obviously satisfied.
\end{example}

\begin{example}
Let
\[
F\left(  \rho_{1},\rho_{2}\right)  :=\mathrm{tr}\,\sqrt{\sqrt{\rho_{1}}%
\rho_{2}\sqrt{\rho_{1}}}\leq1.
\]
Recall
\begin{align}
1-F\left(  \rho_{1},\rho_{2}\right)   &  \leq\frac{1}{2}\left\Vert \rho
_{1}-\rho_{2}\right\Vert _{1}\leq\sqrt{1-F\left(  \rho_{1},\rho_{2}\right)
},\label{1-F<D<1-F}\\
\cos^{-1}F\left(  \rho_{1},\rho_{2}\right)   &  \leq\cos^{-1}F\left(  \rho
_{1}^{\prime},\rho_{2}\right)  +\cos^{-1}F\left(  \rho_{1},\rho_{1}\right)  .
\label{triangle}%
\end{align}
By (\ref{triangle}),
\begin{align*}
F\left(  \rho_{1},\rho_{2}\right)   &  \geq\cos(\cos^{-1}F\left(  \rho
_{1}^{\prime},\rho_{2}\right)  +\cos^{-1}F\left(  \rho_{1},\rho_{1}^{\prime
}\right)  )\\
&  =F\left(  \rho_{1}^{\prime},\rho_{2}\right)  F\left(  \rho_{1},\rho
_{1}^{\prime}\right)  -\left(  \sin\cos^{-1}F\left(  \rho_{1}^{\prime}%
,\rho_{2}\right)  \right)  \left(  \sin\cos^{-1}\left(  F\left(  \rho_{1}%
,\rho_{1}^{\prime}\right)  \right)  \right) \\
&  \geq F\left(  \rho_{1}^{\prime},\rho_{2}\right)  F\left(  \rho_{1},\rho
_{1}^{\prime}\right)  -\sin\cos^{-1}\left(  F\left(  \rho_{1},\rho_{1}%
^{\prime}\right)  \right) \\
&  =F\left(  \rho_{1}^{\prime},\rho_{2}\right)  F\left(  \rho_{1},\rho
_{1}^{\prime}\right)  -\sqrt{1-\left(  F\left(  \rho_{1},\rho_{1}^{\prime
}\right)  \right)  ^{2}}.
\end{align*}
Therefore, by (\ref{1-F<D<1-F}),
\begin{align*}
F\left(  \rho_{1}^{\prime},\rho_{2}\right)  -F\left(  \rho_{1},\rho
_{2}\right)   &  \leq F\left(  \rho_{1}^{\prime},\rho_{2}\right)  \left\{
1-F\left(  \rho_{1},\rho_{1}^{\prime}\right)  \right\}  +\sqrt{1-\left(
F\left(  \rho_{1},\rho_{1}^{\prime}\right)  \right)  ^{2}}\\
&  \leq\frac{1}{2}\left\Vert \rho_{1}-\rho_{1}^{\prime}\right\Vert _{1}%
+\sqrt{2\cdot\frac{1}{2}\left\Vert \rho_{1}-\rho_{1}^{\prime}\right\Vert _{1}%
}.
\end{align*}
Exchanging $\rho_{1}$ and $\rho_{1}^{\prime}$, \
\[
F\left(  \rho_{1},\rho_{2}\right)  -F\left(  \rho_{1}^{\prime},\rho
_{2}\right)  \leq\frac{1}{2}\left\Vert \rho_{1}-\rho_{1}^{\prime}\right\Vert
_{1}+\sqrt{\left\Vert \rho_{1}-\rho_{1}^{\prime}\right\Vert _{1}}.
\]
Therefore,
\[
\left\vert F\left(  \rho_{1},\rho_{2}\right)  -F\left(  \rho_{1}^{\prime}%
,\rho_{2}\right)  \right\vert \leq\frac{1}{2}\left\Vert \rho_{1}-\rho
_{1}^{\prime}\right\Vert _{1}+\sqrt{\left\Vert \rho_{1}-\rho_{1}^{\prime
}\right\Vert _{1}}.
\]
By the symmetry $F\left(  \rho_{1},\rho_{2}\right)  =F\left(  \rho_{2}%
,\rho_{1}\right)  $, we have an analogous upper bound to $\left\vert F\left(
\rho_{1}^{\prime},\rho_{2}\right)  -F\left(  \rho_{1}^{\prime},\rho
_{2}^{\prime}\right)  \right\vert $. Therefore,
\begin{align*}
&  \left\vert F\left(  \rho_{1},\rho_{2}\right)  -F\left(  \rho_{1}^{\prime
},\rho_{2}^{\prime}\right)  \right\vert \\
&  \leq\left\vert F\left(  \rho_{1},\rho_{2}\right)  -F\left(  \rho
_{1}^{\prime},\rho_{2}\right)  \right\vert +\left\vert F\left(  \rho
_{1}^{\prime},\rho_{2}\right)  -F\left(  \rho_{1}^{\prime},\rho_{2}^{\prime
}\right)  \right\vert \\
&  \leq\frac{1}{2}\left\Vert \rho_{1}-\rho_{1}^{\prime}\right\Vert _{1}%
+\sqrt{\left\Vert \rho_{1}-\rho_{1}^{\prime}\right\Vert _{1}}+\frac{1}%
{2}\left\Vert \rho_{2}-\rho_{2}^{\prime}\right\Vert _{1}+\sqrt{\left\Vert
\rho_{2}-\rho_{2}^{\prime}\right\Vert _{1}}\,.
\end{align*}
Thus, \
\[
D_{F}\left(  \rho_{1},\rho_{2}\right)  :=1-F\left(  \rho_{1},\rho_{2}\right)
\]
satisfies the premise of Theorem\thinspace\ref{th:dDconverge}.
\end{example}

\begin{example}
Let
\[
D_{\alpha}\left(  \rho_{1},\rho_{2}\right)  :=\frac{4}{1-\alpha^{2}}\left(
1-\mathrm{tr}\,\rho_{1}^{\frac{1-\alpha}{2}}\rho_{2}^{\frac{1+\alpha}{2}%
}\right)  ,\,\,\left(  -1<\alpha<1\right)
\]
which is monotone decreasing by CPTP maps \cite{Petz:1986}. By Lemmas
\ref{lem:ta}-\ref{lem:f-f<f},
\begin{align*}
\left\Vert \rho_{1}^{\frac{1+\alpha}{2}}-\left(  \rho_{1}^{\prime}\right)
^{\frac{1+\alpha}{2}}\right\Vert  &  \leq\left\Vert \rho_{1}-\rho_{1}^{\prime
}\right\Vert ^{\frac{1+\alpha}{2}}\\
&  \leq\left\Vert \rho_{1}-\rho_{1}^{\prime}\right\Vert _{1}^{\frac{1+\alpha
}{2}}.
\end{align*}
Therefore,
\begin{align*}
\left\vert D_{\alpha}\left(  \rho_{1},\rho_{2}\right)  -D_{\alpha}\left(
\rho_{1}^{\prime},\rho_{2}\right)  \right\vert  &  =\left\vert \mathrm{tr}%
\,\left\{  \rho_{1}^{\frac{1-\alpha}{2}}-\left(  \rho_{1}^{\prime}\right)
^{\frac{1+\alpha}{2}}\right\}  \rho_{2}^{\frac{1+\alpha}{2}}\right\vert \\
&  \leq\,\left\Vert \rho_{1}^{\frac{1-\alpha}{2}}-\left(  \rho_{1}^{\prime
}\right)  ^{\frac{1+\alpha}{2}}\right\Vert \left\Vert \rho_{2}^{\frac
{1+\alpha}{2}}\right\Vert _{1}\\
&  \leq\,\left\Vert \rho_{1}-\rho_{1}^{\prime}\right\Vert _{1}^{\frac
{1+\alpha}{2}}\mathrm{tr}\,\rho_{2}^{\frac{1+\alpha}{2}}.
\end{align*}
Similarly,
\[
\left\vert D_{\alpha}\left(  \rho_{1}^{\prime},\rho_{2}\right)  -D_{\alpha
}\left(  \rho_{1}^{\prime},\rho_{2}^{\prime}\right)  \right\vert
\leq\left\Vert \rho_{2}-\rho_{2}^{\prime}\right\Vert _{1}^{\frac{1-\alpha}{2}%
}\mathrm{tr}\,\left(  \rho_{1}^{\prime}\right)  ^{\frac{1+\alpha}{2}}.
\]
Therefore,
\begin{align*}
&  \left\vert D_{\alpha}\left(  \rho_{1},\rho_{2}\right)  -D_{\alpha}\left(
\rho_{1}^{\prime},\rho_{2}^{\prime}\right)  \right\vert \\
&  \leq\left\vert D_{\alpha}\left(  \rho_{1},\rho_{2}\right)  -D_{\alpha
}\left(  \rho_{1}^{\prime},\rho_{2}\right)  \right\vert +\left\vert D_{\alpha
}\left(  \rho_{1}^{\prime},\rho_{2}\right)  -D_{\alpha}\left(  \rho
_{1}^{\prime},\rho_{2}^{\prime}\right)  \right\vert \\
&  \leq d\left(  \left\Vert \rho_{1}-\rho_{1}^{\prime}\right\Vert _{1}%
^{\frac{1+\alpha}{2}}+\left\Vert \rho_{2}-\rho_{2}^{\prime}\right\Vert
_{1}^{\frac{1-\alpha}{2}}\right)  ,
\end{align*}
and the premise of Theorem\thinspace\ref{th:dDconverge} is satisfied.
\end{example}

\begin{example}
Let $D\left(  \rho_{1},\rho_{2}\right)  $ be a two point functional satisfying
the premise of Theorem\thinspace\ref{th:dDconverge}. Let us define
\[
D_{k}\left(  \rho_{1},\rho_{2},\cdots,\rho_{k}\right)  :=\sum_{i,j=1}%
^{k}a_{ij}D\left(  \rho_{i},\rho_{j}\right)  ,\,
\]
where $a_{ij}\geq0$. Then, $D_{k}$ satisfies the premise of Theorem\thinspace
\ref{th:dDconverge}.
\end{example}

\begin{example}
Let $D\left(  \rho_{1},\rho_{2}\right)  $ be a two point functional satisfying
the premise of Theorem\thinspace\ref{th:dDconverge}. Let us define%
\[
\overline{D_{k}}\left(  \rho_{1},\rho_{2},\cdots,\rho_{k}\right)
:=\inf_{\overline{\rho}\in\mathcal{S}\left(  \mathcal{H}\right)  }\max_{1\leq
j\leq k}D\left(  \rho_{j},\overline{\rho}\right)  .
\]
Then,
\begin{align*}
&  \overline{D_{k}}\left(  \rho_{1},\rho_{2},\cdots,\rho_{k}\right)
-\overline{D_{k}}\left(  \rho_{1}^{\prime},\rho_{2}^{\prime},\cdots,\rho
_{k}^{\prime}\right)  \\
&  =\inf_{\overline{\rho}\in\mathcal{S}\left(  \mathcal{H}\right)  }%
\sup_{\overline{\rho}^{\prime}\in\mathcal{S}\left(  \mathcal{H}\right)  }%
\max_{1\leq j\leq k}\min_{1\leq j^{\prime}\leq k}\left\{  D\left(  \rho
_{j},\overline{\rho}\right)  -D\left(  \rho_{j^{\prime}}^{\prime}%
,\overline{\rho}^{\prime}\right)  \right\}  .
\end{align*}
Therefore, letting $\overline{\rho}_{\varepsilon}\in\mathcal{S}\left(
\mathcal{H}\right)  $ be a state with
\[
\max_{1\leq j^{\prime}\leq k}D\left(  \rho_{j^{\prime}}^{\prime}%
,\overline{\rho}_{\varepsilon}\right)  \leq\inf_{\overline{\rho}\in
\mathcal{S}\left(  \mathcal{H}\right)  }\max_{1\leq j^{\prime}\leq k}D\left(
\rho_{j^{\prime}}^{\prime},\overline{\rho}\right)  +\varepsilon,
\]
we have
\begin{align*}
&  \overline{D_{k}}\left(  \rho_{1},\rho_{2},\cdots,\rho_{k}\right)
-\overline{D_{k}}\left(  \rho_{1}^{\prime},\rho_{2}^{\prime},\cdots,\rho
_{k}^{\prime}\right)  \\
&  \leq\max_{1\leq j\leq k}\min_{1\leq j^{\prime}\leq k}\left\{  D\left(
\rho_{j},\overline{\rho}_{\varepsilon}\right)  -D\left(  \rho_{j^{\prime}%
}^{\prime},\overline{\rho}_{\varepsilon}\right)  \right\}  +\varepsilon\\
&  \leq\max_{1\leq j\leq k}\left\{  D\left(  \rho_{j},\overline{\rho
}_{\varepsilon}\right)  -D\left(  \rho_{j}^{\prime},\overline{\rho
}_{\varepsilon}\right)  \right\}  +\varepsilon\\
&  \leq\sum_{j=1}^{k}\left\{  D\left(  \rho_{j},\overline{\rho}_{\varepsilon
}\right)  -D\left(  \rho_{j}^{\prime},\overline{\rho}_{\varepsilon}\right)
\right\}  +\varepsilon\\
&  \leq\sum_{j=1}^{k}f\left(  \left\Vert \rho_{j}-\rho_{j}^{\prime}\right\Vert
_{1},0\right)  +\varepsilon.
\end{align*}
Since $\varepsilon>0$ is arbitrary, \
\begin{align*}
&  \overline{D_{k}}\left(  \rho_{1},\rho_{2},\cdots,\rho_{k}\right)
-\overline{D_{k}}\left(  \rho_{1}^{\prime},\rho_{2}^{\prime},\cdots,\rho
_{k}^{\prime}\right)  \\
&  \leq\sum_{j=1}^{k}f\left(  \left\Vert \rho_{j}-\rho_{j}^{\prime}\right\Vert
_{1},0\right)  .
\end{align*}
Almost analogously, we also have
\begin{align*}
&  \overline{D_{k}}\left(  \rho_{1},\rho_{2},\cdots,\rho_{k}\right)
-\overline{D_{k}}\left(  \rho_{1}^{\prime},\rho_{2}^{\prime},\cdots,\rho
_{k}^{\prime}\right)  \\
&  \geq-\sum_{j=1}^{k}f\left(  \left\Vert \rho_{j}-\rho_{j}^{\prime
}\right\Vert _{1},0\right)  .
\end{align*}
Therefore, \ $\overline{D_{k}}$ satisfies the premise of \ Theorem\thinspace
\ref{th:dDconverge}.
\end{example}

\section{Classical Markov chains over arbitrary measurable space}

We had shown the convergence of quantum Markov chain in case of finite
dimensional Hilbert space. Next target maybe the analogous statement for
infinite dimensional Hilbert space. Since this is very difficult, instead, we
study classcal Markov chain, but over the arbitrary measurable space.  

Inhomogeneous classical Markov chain with initial probability measure $P$ is
defined by
\[
P_{1}:=\Gamma_{1}\left(  P\right)  ,\,P_{2}:=\Gamma_{2}\left(  P_{1}\right)
,\cdots,
\]
where $P_{t}$ is a probability measure in measureable space $\left(
X,\mathfrak{X}\right)  $, and $\Gamma_{i}$ is a positive linear map from the
space $ca\left(  X,\mathfrak{X}\right)  $ of bounded signed measures over
$\left(  X,\mathfrak{X}\right)  $ to $ca\left(  X,\mathfrak{X}\right)  $, such
that $\left\Vert \Gamma_{i}\left(  \mu\right)  \right\Vert _{1}=$ $\left\Vert
\mu\right\Vert _{1}$ for any positive element $\mu$ of $ca\left(
X,\mathfrak{X}\right)  $.

Consider familis of probability measures $\mathcal{E}:=\left\{  P_{\theta
};\theta\in\Theta\right\}  $ over measurable space $\left(  X,\mathfrak{X}%
\right)  $ and $\mathcal{F}:=\left\{  Q_{\theta};\theta\in\Theta\right\}  $
over $\left(  Y,\mathfrak{Y}\right)  $. Then the relations $\succeq$, $\equiv
$, and two point functions $\Delta$, $\delta$, are defined in analogy to the
ones in Section\thinspace\ref{sec:equivalence-class}. They also satisfy
(\ref{Delta=0-equiv})-(\ref{d>dL}). The equivalence relation $\equiv$ induces
an equvalence class of familis of probability meaures. \ We denote by $\left[
\mathcal{E}\right]  $ the equivalence class to which $\mathcal{E}$ belongs,
and $\mathbb{E}\left(  \Theta\right)  $ denotes the set of alll equvalence
classes of probability distribution families parameterized by elements of
$\Theta$. (That $\mathbb{E}\left(  \Theta\right)  $ is a set and not a proper
class is known\thinspace\cite{LeCam:86}.) For each $\Theta_{0}\subset\Theta$,
define $\mathcal{E}_{\Theta_{0}}:=\left\{  P_{\theta};\theta\in\Theta
_{0}\right\}  $, and denote by $\Pi_{\Theta_{0}}$ the map which sends $\left[
\mathcal{E}\right]  \mathcal{\in\,}\mathbb{E}\left(  \Theta\right)  $ to
$\left[  \mathcal{E}_{\Theta_{0}}\right]  \in\mathbb{E}\left(  \Theta
_{0}\right)  $. Suppose $\Theta_{0}$ is a finite set, and furnish
$\mathbb{E}\left(  \Theta_{0}\right)  $ with the \textit{strong topology},
i.e., the topology induced by the distance $\Delta$. Then, the \textit{weak
topology} of $\mathbb{E}\left(  \Theta\right)  $ is the coarsest topology
which makes $\Pi_{\Theta_{0}}$ continuous for each finite subset $\Theta
_{0}\subset\Theta$. Put differently, the base of the weak topology is in the
form of
\begin{equation}
\bigcap_{\kappa\in K}\Pi_{\Theta_{\kappa}}^{-1}\left(  U_{\kappa}\right)
\label{E-base}%
\end{equation}
where each $\Theta_{\kappa}$ $\subset\Theta$ is a finite subset of $\Theta$,
$K$ is a set of indeces with $\left\vert K\right\vert <\infty$, and each
$U_{\kappa}$ is an open set in $\mathbb{E}\left(  \Theta_{\kappa}\right)  $.

\begin{lemma}
\label{lem:E-compact}\cite{LeCam:86}Let $\Theta$ be a set. Then $\mathbb{E}%
\left(  \Theta\right)  $ is a compact Hausdorff space relative to the weak topology.
\end{lemma}

\begin{lemma}
\label{lem:E-converge}A sequence $\left\{  \left[  \mathcal{E}_{i}\right]
\right\}  _{i=0}^{\infty}$ in $\mathbb{E}\left(  \Theta\right)  $ converges to
$\left[  \mathcal{E}_{\infty}\right]  $ if and only if for each finite subset
$\Theta_{0}$, $\left\{  \Pi_{\Theta_{0}}\left(  \left[  \mathcal{E}%
_{i}\right]  \right)  \right\}  _{i=0}^{\infty}$ converges to $\Pi_{\Theta
_{0}}\left(  \left[  \mathcal{E}_{\infty}\right]  \right)  $ relative to the
strong topology.
\end{lemma}

\begin{proof}
Since `only if ' is trivial, we show ` if\thinspace'. We show that, for any
set in the form of (\ref{E-base}), there is $N$ such that
\begin{equation}
\left\{  \left[  \mathcal{E}_{i}\right]  \right\}  _{i=N}^{\infty}%
\subset\bigcap_{\kappa\in K}\Pi_{\Theta_{\kappa}}^{-1}\left(  U_{\kappa
}\right)  , \label{E-convergence}%
\end{equation}
assuming that $\left\{  \Pi_{\Theta_{\kappa}}\left(  \left[  \mathcal{E}%
_{i}\right]  \right)  \right\}  _{i=0}^{\infty}$ converges to $\Pi
_{\Theta_{\kappa}}\left(  \left[  \mathcal{E}_{\infty}\right]  \right)  $ for
each $\kappa\in K$. By assumption, for any open set $U_{\kappa}$ in
$\mathbb{E}\left(  \Theta_{\kappa}\right)  $, there is $N_{\kappa}$ such that
\[
\left\{  \Pi_{\Theta_{\kappa}}\left(  \left[  \mathcal{E}_{i}\right]  \right)
\right\}  _{i=N_{\kappa}}^{\infty}\subset U_{\kappa}\text{.}%
\]
Therefore,
\[
\left\{  \left[  \mathcal{E}_{i}\right]  \right\}  _{i=N_{^{\kappa}}}^{\infty
}\subset\Pi_{\Theta_{\kappa}}^{-1}\left(  U_{\kappa}\right)  .
\]
Therefore, setting $N:=\max_{\kappa\in K}N_{\kappa}$, we have
\ref{E-convergence}. Thus, the proof is complete.
\end{proof}

We consider families of probability measures
\begin{align*}
\mathcal{E}_{0}  &  :=\mathcal{E=}\left\{  P_{\theta};\theta\in\Theta\right\}
,\\
\mathcal{E}_{1}  &  :=\left\{  P_{\theta,1};\theta\in\Theta\right\}
,P_{\theta,1}=\Gamma_{1}\left(  P_{\theta,0}\right)  ,\\
\mathcal{E}_{2}  &  :=\left\{  P_{\theta,2};\theta\in\Theta\right\}
,P_{\theta,2}=\Gamma_{2}\left(  P_{\theta,1}\right)  ,\cdots,
\end{align*}
and son on.

\begin{theorem}
Let $\left\{  \mathcal{E}_{i}\right\}  _{i=0}^{\infty}$ be as above and
$\Theta$ be any set. Then, for any finite subset $\Theta_{0}$, we have
\begin{equation}
\lim_{i\rightarrow\infty}\Delta\left(  \Pi_{\Theta_{0}}\left(  \left[
\mathcal{E}_{i}\right]  \right)  ,\Pi_{\Theta_{0}}\left(  \left[
\mathcal{E}_{\infty}\right]  \right)  \right)  =0. \label{strong-conv-sub}%
\end{equation}
or equivalently, the sequence $\left\{  \left[  \mathcal{E}_{i}\right]
\right\}  _{i=0}^{\infty}$ in $\mathbb{E}\left(  \Theta\right)  $ converges
relative to the weak topology.
\end{theorem}

\begin{proof}
Since $\mathbb{E}\left(  \Theta\right)  $ is compact by Lemma\thinspace
\ref{lem:E-compact}, there is an accumulation point $\left[  \mathcal{E}%
_{\infty}\right]  $ of the sequence $\left\{  \left[  \mathcal{E}_{i}\right]
\right\}  _{i=0}^{\infty}$. Since $\Pi_{\Theta_{0}}$ is continuous for each
finite set $\Theta_{0}$, $\Pi_{\Theta_{0}}\left(  \left[  \mathcal{E}_{\infty
}\right]  \right)  $ is an accumulation point of \ $\left\{  \Pi_{\Theta_{0}%
}\left(  \left[  \mathcal{E}_{i}\right]  \right)  \right\}  _{i=0}^{\infty}$.
Since $\mathbb{E}\left(  \Theta_{0}\right)  $ satisfies the first axiom
countability, there is a subsequence $\left\{  n_{i}\right\}  _{i=0}^{\infty}$
such that
\[
\lim_{i\rightarrow\infty}\Delta\left(  \Pi_{\Theta_{0}}\left(  \left[
\mathcal{E}_{n_{i}}\right]  \right)  ,\Pi_{\Theta_{0}}\left(  \left[
\mathcal{E}_{\infty}\right]  \right)  \right)  =0.
\]
For any $i_{1}\leq i_{2}$%
\[
\delta\left(  \Pi_{\Theta_{0}}\left(  \left[  \mathcal{E}_{i_{1}}\right]
\right)  ,\Pi_{\Theta_{0}}\left(  \left[  \mathcal{E}_{i_{2}}\right]  \right)
\right)  =0.
\]
Therefore, taking $j$ with $n_{j}\leq i$ ,
\begin{align*}
\delta\left(  \Pi_{\Theta_{0}}\left(  \left[  \mathcal{E}_{\infty}\right]
\right)  ,\Pi_{\Theta_{0}}\left(  \left[  \mathcal{E}_{i}\right]  \right)
\right)   &  \leq\delta\left(  \Pi_{\Theta_{0}}\left(  \left[  \mathcal{E}%
_{\infty}\right]  \right)  ,\Pi_{\Theta_{0}}\left(  \left[  \mathcal{E}%
_{n_{j}}\right]  \right)  \right)  +\delta\left(  \Pi_{\Theta_{0}}\left(
\left[  \mathcal{E}_{n_{j}}\right]  \right)  ,\Pi_{\Theta_{0}}\left(  \left[
\mathcal{E}_{i}\right]  \right)  \right) \\
&  =\delta\left(  \Pi_{\Theta_{0}}\left(  \left[  \mathcal{E}_{\infty}\right]
\right)  ,\Pi_{\Theta_{0}}\left(  \left[  \mathcal{E}_{n_{j}}\right]  \right)
\right) \\
&  \leq\Delta\left(  \Pi_{\Theta_{0}}\left(  \left[  \mathcal{E}_{n_{j}%
}\right]  \right)  ,\Pi_{\Theta_{0}}\left(  \left[  \mathcal{E}_{\infty
}\right]  \right)  \right)  \rightarrow0,\text{ }i\rightarrow\infty
,n_{j}\rightarrow\infty.
\end{align*}
Also, taking $j$ with $n_{j}\geq i,$
\begin{align*}
\delta\left(  \Pi_{\Theta_{0}}\left(  \left[  \mathcal{E}_{i}\right]  \right)
,\Pi_{\Theta_{0}}\left(  \left[  \mathcal{E}_{\infty}\right]  \right)
\right)   &  \leq\delta\left(  \Pi_{\Theta_{0}}\left(  \left[  \mathcal{E}%
_{i}\right]  \right)  ,\Pi_{\Theta_{0}}\left(  \left[  \mathcal{E}_{n_{j}%
}\right]  \right)  \right)  +\delta\left(  \Pi_{\Theta_{0}}\left(  \left[
\mathcal{E}_{n_{j}}\right]  \right)  ,\Pi_{\Theta_{0}}\left(  \left[
\mathcal{E}_{\infty}\right]  \right)  \right) \\
&  =\delta\left(  \Pi_{\Theta_{0}}\left(  \left[  \mathcal{E}_{n_{j}}\right]
\right)  ,\Pi_{\Theta_{0}}\left(  \left[  \mathcal{E}_{\infty}\right]
\right)  \right) \\
&  \leq\Delta\left(  \Pi_{\Theta_{0}}\left(  \left[  \mathcal{E}_{n_{j}%
}\right]  \right)  ,\Pi_{\Theta_{0}}\left(  \left[  \mathcal{E}_{\infty
}\right]  \right)  \right)  \rightarrow0,\text{ }n_{j}\rightarrow\infty.
\end{align*}
Therefore, we have (\ref{strong-conv-sub}).
\end{proof}

\bigskip

We say a Markov process is \textit{weakly ergodic} if and only if
\[
\lim_{i\rightarrow\infty}\sup_{P,P^{\prime}}\left\Vert \Gamma_{i}\circ
\cdots\circ\Gamma_{2}\circ\Gamma_{1}\left(  P\right)  -\Gamma_{i}\circ
\cdots\circ\Gamma_{2}\circ\Gamma_{1}\left(  P^{\prime}\right)  \right\Vert
_{1}=0,
\]
and $L^{1}$-\textit{weakly ergodic} if and only if
\[
\forall P_{,}P^{\prime},\,\,\lim_{i\rightarrow\infty}\left\Vert \Gamma
_{i}\circ\cdots\circ\Gamma_{2}\circ\Gamma_{1}\left(  P\right)  -\Gamma
_{i}\circ\cdots\circ\Gamma_{2}\circ\Gamma_{1}\left(  P^{\prime}\right)
\right\Vert _{1}=0
\]

\begin{theorem}
Let $\mathcal{E}_{\ast}$ be a state family such that
\[
\mathcal{E}_{\ast}=\left\{  P_{\ast};\theta\in\Theta\right\}  .
\]
A Markov process is weakly ergodic if and only if, for any $\mathcal{E}$,
$\left\{  \left[  \mathcal{E}_{i}\right]  \right\}  _{i=1}^{\infty}%
$\ converges to $\left[  \mathcal{E}_{\ast}\right]  $ relative to the strong
topology. Also. a Markov process is $L_{1}$-weakly ergodic if and only if, for
any $\mathcal{E}$, $\left\{  \left[  \mathcal{E}_{i}\right]  \right\}
_{i=1}^{\infty}$\ converges to $\left[  \mathcal{E}_{\ast}\right]  $ relative
to the weak topology.
\end{theorem}

\begin{proof}
The first statement is proved by the argument almost parallel to the proof of
Theorem\thinspace\ref{th:ergodic}. Thus we only prove the second statement.
Suppose the Markov chain is $L_{1}$-weakly ergodic. Fix a finite subset
$\Theta_{0}\subset\Theta$, and $\theta_{0}\in\Theta_{0}$ .Let $\Lambda_{i}$ be
a CPTP map with $\Lambda_{i}\left(  P_{\ast}\right)  =P_{\theta_{0},i}$ .
Then,
\begin{align*}
&  \sup_{\theta\in\Theta_{0}}\left\Vert P_{\theta,i}-\Lambda_{i}\left(
P_{\ast}\right)  \right\Vert _{1}\\
&  \leq\sum_{\theta\in\Theta_{0}}\left\Vert P_{\theta,i}-\Lambda_{i}\left(
P_{\ast}\right)  \right\Vert _{1}\\
&  =\sum_{\theta\in\Theta_{0}}\left\Vert P_{\theta,i}-P_{\theta_{0}%
,i}\right\Vert _{1}\rightarrow0,
\end{align*}
or equivalently,
\[
\lim_{i\rightarrow\infty}\delta\left(  \Pi_{\Theta_{0}}\left(  \mathcal{E}%
_{\ast}\right)  ,\Pi_{\Theta_{0}}\left(  \mathcal{E}_{i}\right)  \right)  =0.
\]
On the other hand, denoting by $\Lambda_{\ast}$ the CPTP map with
$\Lambda_{\ast}\left(  \rho\right)  =\rho_{\ast}$ for any $\rho_{\ast}$,
\[
\delta\left(  \Pi_{\Theta_{0}}\left(  \mathcal{E}_{i}\right)  ,\Pi_{\Theta
_{0}}\left(  \,\mathcal{E}_{\ast}\right)  \right)  \leq\sup_{\theta\in
\Theta_{0}}\left\Vert \Lambda_{\ast}\left(  \rho_{\theta,i}\right)
-\rho_{\ast}\right\Vert _{1}=0.
\]
Therefore, we have
\[
\lim_{i\rightarrow\infty}\Delta\left(  \Pi_{\Theta_{0}}\left(  \mathcal{E}%
_{\ast}\right)  ,\Pi_{\Theta_{0}}\left(  \mathcal{E}_{i}\right)  \right)  =0.
\]
Since this is true for any finite subset $\Theta_{0}\subset\Theta$, by
Lemma\thinspace\ref{lem:E-converge}, $\left\{  \left[  \mathcal{E}_{i}\right]
\right\}  _{i=1}^{\infty}$\ converges to $\left[  \mathcal{E}_{\ast}\right]  $
relative to the weak topology.

Next, suppose that, for any $\mathcal{E}$, $\left\{  \left[  \mathcal{E}%
_{i}\right]  \right\}  _{i=1}^{\infty}$\ converges to $\left[  \mathcal{E}%
_{\ast}\right]  $ relative to the weak topology. Especially, let
$\Theta=\left\{  1,2\right\}  $ and $P_{1}=P,\,P_{2}=P^{\prime}$. Then, since
$\Theta$ is finite set, we have
\[
\lim_{i\rightarrow\infty}\Delta\left(  \mathcal{E}_{\ast},\mathcal{E}%
_{i}\right)  \geq\lim_{i\rightarrow\infty}\delta\left(  \mathcal{E}_{\ast
},\mathcal{E}_{i}\right)  =0
\]
Then, we have
\begin{align*}
&  \left\Vert \Gamma_{i}\circ\cdots\circ\Gamma_{2}\circ\Gamma_{1}\left(
P\right)  -\Gamma_{i}\circ\cdots\circ\Gamma_{2}\circ\Gamma_{1}\left(
P^{\prime}\right)  \right\Vert _{1}\\
&  =\left\Vert P_{1,i}-P_{2,i}\right\Vert _{1}\leq\left\Vert P_{1,i}%
-\Lambda_{i}\left(  P_{\ast}\right)  \right\Vert _{1}+\left\Vert
P_{2,i}-\Lambda_{i}\left(  P_{\ast}\right)  \right\Vert _{1}\\
&  \leq\sup_{\theta\in\Theta}\left\Vert P_{\theta,i}-\Lambda_{i}\left(
P_{\ast}\right)  \right\Vert _{1}%
\end{align*}
Since this holds for any $\Lambda_{i}$, we have
\begin{align*}
&  \left\Vert \Gamma_{i}\circ\cdots\circ\Gamma_{2}\circ\Gamma_{1}\left(
P\right)  -\Gamma_{i}\circ\cdots\circ\Gamma_{2}\circ\Gamma_{1}\left(
P^{\prime}\right)  \right\Vert \\
&  \leq\inf_{\Lambda_{i}}\sup_{\theta\in\Theta}\left\Vert P_{\theta,i}%
-\Lambda_{i}\left(  P_{\ast}\right)  \right\Vert _{1}=\delta\left(
\mathcal{E}_{\ast},\mathcal{E}_{i}\right)  \rightarrow0.
\end{align*}
Therefore, the Markov chain is $L_{1}$-weakly ergodic.
\end{proof}

\begin{theorem}
Suppose $\Gamma_{i}=\Gamma$ ($i=1,2,\cdots$). Then,
\[
\Gamma\left(  \mathcal{E}_{\infty}\right)  \equiv\mathcal{E}_{\infty}.
\]

\end{theorem}

\begin{proof}
By (\ref{strong-conv-subfamily}) and using the almost parallel argument as the
proof of Theorem \ref{th:fixed-point} leads to%
\[
\Pi_{\Theta_{0}}\left(  \Gamma\left(  \mathcal{E}_{\infty}\right)  \right)
=\Gamma\left(  \Pi_{\Theta_{0}}\left(  \mathcal{E}_{\infty}\right)  \right)
=\Pi_{\Theta_{0}}\left(  \mathcal{E}_{\infty}\right)  .
\]
Observe $\mathbb{E}\left(  \Theta\right)  $ is a Hausforff space relative to
weak topology by Lemma\thinspace\ref{lem:E-compact}. Therefore, if
$\Gamma\left(  \mathcal{E}_{\infty}\right)  \not \equiv \mathcal{E}_{\infty}%
$,
\begin{align}
\Gamma\left(  \mathcal{E}_{\infty}\right)   &  \in\bigcap_{\kappa\in K}%
\Pi_{\Theta_{\kappa}}^{-1}\left(  U_{\kappa}\right)  ,\,\mathcal{E}_{\infty
}\in\bigcap_{\kappa\in K^{\prime}}\Pi_{\Theta_{\kappa}}^{-1}\left(  U_{\kappa
}^{\prime}\right)  ,\,\label{include}\\
&  \bigcap_{\kappa\in K}\Pi_{\Theta_{\kappa}}^{-1}\left(  U_{\kappa}\right)
\cap\bigcap_{\kappa\in K^{\prime}}\Pi_{\Theta_{\kappa}}^{-1}\left(  U_{\kappa
}^{\prime}\right)  \label{no-overlap}%
\end{align}
holds where $\Theta_{\kappa}$ is a finite subset of $\Theta$, \ and $K$,
$K^{\prime}$ is a finite set of indeces, and each $U_{\kappa}$, $U_{\kappa
}^{\prime}$ is an open set in $\mathbb{E}\left(  \Theta_{\kappa}\right)  $.
For (\ref{no-overlap}) to hold, it is necessary that $K$ and $K^{\prime}$
share at least one element $\kappa_{0}$. Also, it is necessary that
\[
U_{\kappa_{0}}\cap U_{\kappa_{0}}^{\prime}=\emptyset.
\]
By (\ref{include}), we have to have $\Pi_{\Theta_{0}}\left(  \Gamma\left(
\mathcal{E}_{\infty}\right)  \right)  \in U_{\kappa_{0}}$ and $\Pi_{\Theta
_{0}}\left(  \mathcal{E}_{\infty}\right)  \in U_{\kappa_{0}}^{\prime}$ . Since
$\mathbb{E}\left(  \Theta_{\kappa}\right)  $ is a Hausforff space, this means
\
\[
\Pi_{\Theta_{0}}\left(  \Gamma\left(  \mathcal{E}_{\infty}\right)  \right)
\neq\Pi_{\Theta_{0}}\left(  \mathcal{E}_{\infty}\right)  .
\]
This contradicts with the assumption. Therefore, we have to have
$\Gamma\left(  \mathcal{E}_{\infty}\right)  \equiv\mathcal{E}_{\infty}$.
\end{proof}

\section{Discussions}

We had found out that any quantum Markov chain "converges", if you introduce
a proper equivalence class. This equivalence class, as had been pointed out,
has good decision theoretic meaning \cite{Matsumoto}. The mode of convergence
in case of finite dimensional Hilbert space is "strong convergence", that is,
convergence with respect to the metric $\Delta$ (\ref{converge}). But, even
for the classical Markov chains, such strong statement does not hold in
general.  Instead, what we could prove was weak convergence
(\ref{strong-conv-sub}). Hence, also in quantum case, this is what we can
expect at most. The author conjecture weak convergence (\ref{strong-conv-sub})
holds for any quantum Markov chains over  arbitrary Hilbert spaces. 

Also, we could characterize weak ergodicity and $L^{1}$-weak ergodicity in
view of convergence of the state family, in case of finite dimensional quantum
systems and arbitrary classical systems. The author conjectures the similar
assertion should holds for arbitary quantum systems.

\appendix  

\section{General topology}

A set in a topological space is a \textit{neighborhood} of a point $x$ if and
only if the set contains an open set to which $s$ belongs. A \textit{base} of
the topology is a family of open sets such that for each point $s$, every
neighborhood of $s$ contains a member of the family. A topological space
satisfies the \textit{first axiom of countability} if the topology has a
countable base.

\begin{lemma}
\label{lem:metric-countable}(11, Chapter\thinspace4 of \cite{kelley}) Every
pseudo-metric space satisfies the first axiom of countability.
\end{lemma}

A point $s$ is an \textit{accumulation point} of a subset $A$ of a topological
space if and only if every neighborhood of $s$ contains points of $A$ other
than $s$.

\begin{lemma}
\label{lem:cluster}(8, Chapter 2 of \cite{kelley}) Suppose the first axiom of
countability is satisfied. Then, $s$ is an accumulation point of a sequence
$S$ if and only if there is a subsequence converging to $s$.
\end{lemma}

By Theorem\thinspace5.2 of \cite{kelley}, we have:

\begin{lemma}
\label{lem:compact-accumulate}If a topological space $X$ is copmpact, each
sequence $S$ has an accumulation point.
\end{lemma}

\section{Perturbation of matrix functions}

A function $f$ is said to be operator monotone if and only if $f\left(
A\right)  \geq f\left(  B\right)  $ holds for any Hermitian matrices $A$, $B$
with $A\geq B$.

\begin{lemma}
\label{lem:ta}(Theorem V.1.9 of \cite{Bhatia}) $f\left(  t\right)  =t^{\alpha
}$ ($0\leq\alpha\leq1$) is an operator monotone function on $[0,\infty)$
\end{lemma}

\begin{lemma}
\label{lem:f-f<f}(Theorem X.1.1 of \cite{Bhatia}) Let $f$ be an operator
monotone function on $[0,\infty)$ such that $f\left(  0\right)  =0$. Then, for
all positive operators $A$, $B$,
\[
\left\Vert f\left(  A\right)  -f\left(  B\right)  \right\Vert \leq f\left(
\left\Vert A-B\right\Vert \right)  .
\]

\end{lemma}

\section{}


\begin{thebibliography}{9}                                                                                                %


\bibitem {Bhatia}R.\thinspace Bhatia, "Matrix analysis", Springer (1996)

\bibitem {GutaJencova}M.\thinspace Guta and A.\thinspace Jencova, "Local
Asymptotic Normality in Quantum Statistics", Communications in Mathematical
Physics, vol. 276, No. 2, 341-379 (2007)

\bibitem {kelley}J.\thinspace L.\thinspace Kelley, "General Topology",
Springer (1975).

\bibitem {LeCam:86}L.\thinspace LeCam, "Asymptotic Methods in Statistical
Decision Theory", Springer, New York (1986).

\bibitem {lindqvist:77}B.\thinspace Lindqvist, "How fast does a Markov chain
forget the initial state? A decision theoretical approach", Scand. J. Statist,
4, 145-152 (1977)

\bibitem {lindqvist:78}B.\thinspace Lindqvist, "A decision theoretical
characterization of weak ergodicity", Probability Theory and Related Fields,
Vol. 44, No. 2, 155-158 (1978)

\bibitem {Matsumoto}K.\thinspace Matsumoto, "A quantum version of
randomization criterion", http://arxiv.org/abs/1012.2650.

\bibitem {Petz:1986}D.\thinspace Petz, "Quasi-entropies for finite dimensional
quantum systems", Reports on Mathematical Physics, vol. 23, No.\thinspace1,
57-65 (1986)

\bibitem {Torgersen}E.\thinspace N.\thinspace Torgersen, "Comparison of
experiments when the parameter space is finite," Probability Theory and
Related Fields, Vol. 16, No. 3, 219-249 (1970)
\end{thebibliography}
\end{document}